\documentclass[journal,10pt,final]{IEEEtran}

\usepackage{tikz}
\usetikzlibrary{calc}
\usepackage{cite}
\usepackage{graphicx}
\usepackage{algorithm}
\usepackage{algorithmic}
\usepackage{epsfig}
\usepackage{epstopdf}
\usepackage{amssymb,amsmath}
\usepackage{amsmath}
\usepackage{amsthm}
\usepackage{amsfonts}
\usepackage{subfigure}
\usepackage{psfrag}
\usepackage{rotating}
\usepackage{latexsym}
\usepackage{dsfont}
\usepackage{stmaryrd}
\usepackage{amsthm}
\usepackage{amssymb}
\usepackage{amsmath}

\usepackage[force,almostfull]{textcomp}
\usepackage{fancyvrb}
\usepackage{textcomp}
\usepackage{url}
\usepackage{ifdraft}
\usepackage{url}
\usepackage{multirow}
\usepackage{rotating}
\usepackage{xspace}
\usepackage{array}
\usepackage{multido}
\usepackage{flushend}
 \usepackage{enumitem}
\usepackage{enumerate}
\usepackage[latin1]{inputenc}
\newlength\figureheight 
\newlength\figurewidth 
\usepackage{graphics}
\usepackage{pgfplots}
\pgfplotsset{compat=newest}
\pgfplotsset{plot coordinates/math parser=false}

\usepackage{scalefnt}

\newtheoremstyle{specialcasestyle}{1mm}{1mm}{\upshape}{}{\bfseries\upshape}{.}{0mm}{}
\theoremstyle{specialcasestyle}

\newtheorem{prop}{Proposition}

\begin{document}

\title{An Accurate Sample Rejection Estimator for the Estimation of Outage Probability of EGC Receivers}

\author{Nadhir Ben Rached$^{1}$, Abla Kammoun$^2$, Mohamed-Slim~Alouini$^2$, and Ra\'ul Tempone$^{2,3}$
\\
\thanks{\vspace{-0.2in}\hrule \vspace{0.2cm}
A part of this work has been accepted for publication in IEEE Global Communications Conference (Globecom 2018), Abu Dhabi, UAE, Dec. 2018.

This work was supported by the KAUST Office of Sponsored Research (OSR) under Award No. URF/1/2584-01-01 and the Alexander von Humboldt Foundation.

$^1$ Chair of Mathematics for Uncertainty Quantification, Department of Mathematics, RWTH Aachen University, 52062 Aachen, Germany.

$^2$ Computer, Electrical and Mathematical Sciences \& Engineering Division (CEMSE), King Abdullah University of Science and Technology (KAUST),  23955-6900 Thuwal, Saudi Arabia.

$^3$ Alexander von Humboldt Professor in Mathematics for Uncertainty Quantification, RWTH Aachen University, 52062 Aachen, Germany.

}

}
\date{}
\maketitle
\thispagestyle{empty}
\begin{abstract}
In this work, we evaluate the outage probability (OP) for $L-$branch equal gain combining (EGC) diversity receivers operating over fading channels, i.e. equivalently the cumulative distribution function (CDF) of the sum of the $L$ channel envelopes. In general, closed form expressions of OP values are unobtainable. The use of Monte Carlo (MC) simulations is not considered a good alternative as it requires a large number of samples for small values of OP, making MC simulations very expensive. In this paper, we use the concept of importance sampling (IS), being known to yield accurate estimates using fewer simulation runs. Our proposed IS scheme is essentially based on sample rejection where the IS probability density function (PDF) is the truncation of the underlying PDF over the $L$ dimensional sphere. It assumes the knowledge of the CDF of the sum of the $L$ channel gains in a closed-form expression. Such an assumption is not restrictive since it holds for various challenging fading models. We apply our approach to the case of independent Rayleigh, correlated Rayleigh, and independent and identically distributed Rice fading models. Next, we extend our approach to the interesting scenario of generalised selection combining receivers combined with EGC under the independent Rayleigh fading environment. For each case, we prove the desired bounded relative error property. Finally, we validate these theoretical results through some selected experiments. 
\end{abstract}

\section{Introduction}
Sums of random variables (RVs) occur in many challenging wireless communication applications. For instance, the instantaneous signal-to-noise-ratio (SNR) expressions at the output of equal gain combining (EGC) and maximum ratio combining (MRC) diversity receivers involve sums of RVs \cite{alouini}. Therefore, the evaluation of outage probability (OP) values turns out to be equivalent to computing the cumulative distribution function (CDF) of fading channel envelopes for EGC and of channel gains for MRC \cite{5671659}. Sums of RVs play a central role when the generalised selection combining (GSC) scheme is combined with either EGC or MRC techniques \cite{8472928}. In such cases, the expressions of the OP are given by the CDFs of sums of ordered channel amplitudes for GSC/EGC or channel gains for GSC/MRC. 

Except for the CDF of the sum of two Rayleigh distributions  \cite{1244789}, closed-form expressions of the CDF of the sum of fading channel envelopes have not yet been derived in the literature. To address this knowledge gap, various approximation methods have been proposed. For example, closed-form approximations have been developed for the case of independent Rician fading RVs \cite{1388730,4781943,1421185}. In \cite{1388722}, a simple approximate expression of the CDF of Rayleigh sums was derived. Approximations of the sum of $\kappa-\mu$ and $\eta-\mu$ distributions have also been considered in \cite{4939219}. An extensive interest was devoted to the case of the sum of Log-normal RVs for which various approximation methods have been proposed \cite{4275022,1275712,4814351,1097606,citeulike:7151841}. 

Generally, the accuracy of these closed-form approximations is not always ensured and may degrade for a certain choice of systems parameters. Therefore, alternative approaches are of important practical interest. The Monte Carlo (MC) method presents one alternative method. However, this method requires substantial computational effort when small values of the CDF are considered, thus making this method impractical. To avoid this, variance reduction techniques are used extensively in the context of rare events simulations \cite{rubino2009rare,opac-b1132466}. Importance sampling (IS) is the most popular variance reduction technique and is known, when used appropriately, to yield a very accurate estimate of OP with a fewer number of runs.

There are numerous examples in the literature on the estimation of tail probabilities of sums of RVs using the IS approach. However, few works have been developed on the probability that a sum of RVs is less than a sufficiently small threshold, as we propose here.. For instance, in the Log-normal fading environment, an exponential twisting approach has been proposed in \cite{asmussen2014exponential} to deal with the CDF of independent and identically distributed (i.i.d) sum of Log-normal variates. The correlated Log-normal case has also been considered in \cite{gulisashvili2016,botev_SLN,Nadhir_SLN}. Efficient IS schemes have been developed to estimate the CDF of the sum of Gamma-Gamma \cite{7835220} and $\kappa-\mu$, $\eta-\mu$ and $\alpha-\mu$ \cite{8125718} RVs. In \cite{7328688}, two unified IS approaches have been proposed to estimate OP values over a generalised fading framework using the well-known hazard rate twisting technique \cite{Juneja:2002:SHT:566392.566394,BenRached2016}. Finally, IS and conditional MC (another popular variance reduction technique) estimators have been proposed in \cite{8472928} to estimate the CDF of partial sums of ordered independent RVs that are useful to estimate OP values for GSC/EGC or GSC/MRC receivers.

Contrary to the evaluation of OP under the EGC diversity model, closed-form expressions of OP at the output of MRC diversity receivers are available for many challenging fading environments. This is the case for independent but not necessarily identically distributed (i.n.i.d) Rayleigh fading channels where the expression of OP at the output of MRC receivers is the CDF of the sum of i.n.i.d exponential RVs which is given in \cite{Botev:2013:SNR:2466677.2466683}. The same observation holds for the correlated Rayleigh case \cite{1201072}. The i.i.d $\kappa-\mu$ and $\eta-\mu$ fading models are other examples where the values of OP with the MRC scheme are given respectively by the CDF of the squared $\kappa-\mu$ and squared $\eta-\mu$ variates \cite{4231253}.  A further interesting example is when GSC is combined with MRC under the i.n.i.d Rayleigh fading channels. The OP expression, which is given in this case by the CDF of sums of ordered i.n.i.d exponential variates, is given in closed-form \cite{1033001}. 

These observations provide the main motivation for our study. We propose an IS estimator of the OP at the output of EGC diversity receivers, i.e. the probability that the sum of fading channel envelopes (or the sum of ordered fading channel envelopes in the case of GSC/EGC receivers) falls below a given threshold, based on the knowledge of a closed-form expression of the OP with MRC scheme, i.e. the probability that the sum of channel gains (or the sum of ordered channel gains in the case of GSC/MRC receivers) is less than a certain threshold. More specifically, our proposed IS scheme is based on sample rejection where the biased probability density function (PDF) is given by the truncation of the underlying PDF over the multidimensional hypersphere with a radius equal to the specified threshold. As previously mentioned, assuming the knowledge of a closed-form expression of the OP with MRC scheme is not restrictive since this assumption holds for several practical fading models.
After we explain the general concept of the proposed estimator, we apply our approach to four interesting scenarios, namely the i.n.i.d Rayleigh, the correlated Rayleigh with exponential correlation, the i.i.d Rice, and the i.n.i.d Rayleigh when EGC is combined with GSC. We provide for each case a detailed procedure on how the proposed estimator is implemented and we prove that the bounded relative error property, which is one of the desired properties in the context of rare event simulations \cite{opac-b1132466}, is achieved. Note that in addition to its simplicity in implementation and analysis, the scope of applicability of our proposed IS estimator includes the sum of correlated Rayleigh RVs, which has not yet been considered by other existing approaches. Moreover, although an estimator of the CDF of the sum of i.i.d Rice variates has been developed in \cite{7328688}, it is not clear how sampling according to the biased PDF is performed. This constitutes another contribution of the present work where the CDF of the i.i.d sum of Rice variates is easily implemented. Finally, we compare the performance of our proposed IS estimator through various numerical results with some existing estimators as well as the naive MC sampler. 

The rest of the paper is organised as follows. In Section II, we present the problem setting and describe the main concept of IS. Section III is devoted to presenting the general idea of the proposed IS estimator. Moreover, we apply, in the same section, our IS estimator to four interesting scenarios. For each scenario, we provide a detailed implementation procedure and prove that the desired property of bounded relative error holds. Finally, a comparison of our estimator with some existing estimators as well as naive MC simulations is performed in Section IV. 

\section{Problem Setting}
The instantaneous SNR at the output of $L-$branch EGC diversity receiver is expressed as in \cite{7328688,5671659}
\begin{align}\label{eq_snr}
\gamma_{end}=\frac{E_s}{N_0 L} \left (\sum_{i=1}^{L}{R_i} \right )^2,
\end{align}
where $\frac{E_s}{N_0}$ is the SNR per symbol at the transmitter, $L$ is the number of diversity branches, and $R_i$ is the channel envelope (the fading channel amplitude) of the $i^{th}$ diversity branch. The OP, which is a widely used metric for performance analysis of wireless communication systems operating over fading channels, is defined as the probability that the SNR $\gamma_{end}$ is below a given threshold $\gamma_{th}$
\begin{align}
P_{out}=P\left ( \gamma_{end} \leq \gamma_{th} \right ),
\end{align} 
which is equivalent, using the SNR expression in (\ref{eq_snr}), to 
\begin{align}
P_{out}=P \left ( \sum_{i=1}^{L}{R_i} \leq \gamma_0\right ),
\end{align}
where $\gamma_0=\sqrt {\frac{\gamma_{th} L N_0}{E_s}}$. Thus, the problem is reduced to evaluating the CDF of the sum of fading envelopes (modulus of the fading channels) of the $L$ diversity branches. Unfortunately, this quantity is out of reach for many practical fading models. A non-exhaustive list includes, for instance, the Rayleigh fading environment where the CDF of the sum of correlated (or even independent) Rayleigh RVs is not known to have a closed-form expression. A similar observation also holds for the independent Rician, the $\kappa-\mu$, and the $\eta-\mu$ fading models. Note that when GSC is combined with EGC, the OP expression corresponds to the CDF of partial sums of ordered fading channel amplitudes, i.e. the CDF of the sum of the $N$ largest fading channel amplitudes with $1 \leq N \leq L$.

Naive MC simulations constitute a good alternative to estimate the CDF of the sum of fading envelopes. Let $f(\cdot)$ denote the joint PDF of the random vector containing the $L$ fading envelopes $\bold{R}=(R_1,R_2,\cdots,R_L)$. Then, using $M$ independent replicants $\{ \bold{R}^{(k)}\}_{k=1}^{M}$ of the random vector $\bold{R}$ sampled according to $f(\cdot)$, the naive MC estimator is defined as
\begin{align}
\hat {P}_{out,MC}=\frac{1}{M}\sum_{k=1}^{M}{\bold{1}_{ \left ( \sum_{i=1}^{L}{R_i^{(k)}} \leq \gamma_0\right )}},
\end{align}
where $\bold{1}_{ \left ( \cdot\right )}$ denotes the indicator function. However, the high computational complexity incurred by this method, in terms of required number of samples to ensure an accurate estimate, makes it impractical for sophisticated wireless communication systems where $P_{out}$ is sufficiently small. To illustrate such a point, the naive MC sampler requires a number of runs approximately equal to $100/P_{out}$ to estimate $P_{out}$ with a $20\%$ relative error. 

When appropriately used, IS can save a substantial amount of computational gain compared to naive MC simulations. The concept of IS is to rewrite $P_{out}=\mathbb{E}_f \left [\bold{1}_{\left (\sum_{i=1}^{L}{R_i}  \leq \gamma_0\right)} \right ]$, where $\mathbb{E}_f [\cdot]$ is the expectation with respect to the PDF $f(\cdot)$, as follows
\begin{align}
P_{out}=\mathbb{E}_g \left [ \bold{1}_{\left (\sum_{i=1}^{L}{R_i} \leq \gamma_0\right )} L(R_1,\cdots,R_L)\right ],
\end{align} 
where $g(\cdot)$ is a new PDF named as IS PDF or biased PDF and $\mathbb{E}_g [\cdot]$ denotes the expectation operator with respect to the PDF $g(\cdot)$. $L$ is the likelihood ratio defined as the ratio between the original and the new introduced PDFs
\begin{align}
L(R_1,\cdots,R_L)=\frac{f(R_1,\cdots,R_L)}{g(R_1,\cdots,R_L)}.
\end{align} 
Then, using $M$ samples $\{ \bold{R}^{(k)}\}_{k=1}^{M}$ of the random vector $\bold{R}$ sampled according to $g(\cdot)$, we construct the IS estimator as follows
\begin{align}
\hat {P}_{out,IS}=\frac{1}{M} \sum_{k=1}^{M}{ \bold{1}_{\left (\sum_{i=1}^{L}{R_i^{(k)}} \leq \gamma_0\right )} L(R_1^{(k)},\cdots,R_L^{(k)})}.
\end{align}
The remaining step is the  choice of  biased PDF $g(\cdot)$ that results in a variance reduction and hence in a computational gain with respect to naive MC simulations. Before that, it is necessary to define some performance metrics that serve to measure the goodness of an estimator. Among these criteria, we focus on the bounded relative error property \cite{opac-b1123521}. 
 We say that the estimator $\bold{1}_{\left (\sum_{i=1}^{L}{R_i} \leq \gamma_0\right )} L(R_1,\cdots,R_L)$ achieves the bounded relative error property when
\begin{align}
\limsup_{\gamma_0 \rightarrow 0} {\frac{\mathrm{var}_g \left [\bold{1}_{\left (\sum_{i=1}^{L}{R_i} \leq \gamma_0\right )} L(R_1,\cdots,R_L) \right ]}{P_{out}^2}} < + \infty.
\end{align}
This property has been used, for instance, in \cite{7328688} and implies that , when it holds, the number of samples needed to meet a certain accuracy requirement remains bounded regardless of how small $P_{out}$ is. Hence, it suffices to guarantee a substantial amount of computational gain over naive MC simulations.

\section{Sample Rejection IS Estimator}
Before presenting our choice of the biased PDF $g(\cdot)$, we describe the optimal IS density which is defined as the truncation of $f(\cdot)$ over the rare set $\{\sum_{i=1}^{L}{R_i} \leq \gamma_{0}\}$
\begin{align}
g^*(r_1,\cdots,r_L)=\frac{f(r_1,\cdots,r_L) \bold{1}_{\left (\sum_{i=1}^{L}{r_i} \leq \gamma_0 \right )}}{P_{out}}.
\end{align}
The above optimal IS density, known also as the zero variance measure, is impractical since it involves the unknown quantity $P_{out}$. However, this measure provides some insights on how the IS density may be selected in order to yield a substantial amount of variance reduction. In fact, the optimal IS density encourages samples that belong to the rare set and maintains over it  the likelihood ratio constant. To this end, we propose a biased PDF that is the truncation of the underlying PDF $f(\cdot)$ over a set $S$:
\begin{align}\label{biased_pdf}
g(r_1,\cdots,r_L)=\frac{f(r_1,\cdots,r_L) \bold{1}_{\left (\bold{R} \in S\right )}}{\tilde{P}_{out}},
\end{align}
where $S$ is a set that contains the set of interest $\{(r_1,\cdots,r_L),\sum_{i=1}^{L}{r_i} \leq \gamma_{0},r_i\geq 0 \}$ and $\tilde{P}_{out}$ is the probability that the random vector $\bold{R}$ is in $S$. Obviously, in order to be able to implement the proposed IS approach with the biased PDF above, the quantity $\tilde{P}_{out}$ must be known in  closed-form. 

Our choice of $S$ follows from the following observation. For many fading models with MRC receivers, the OP, which is given in this case by the CDF of the sum of squared fading envelopes, is known in a closed-form expression. This is the case for i.n.i.d Rayleigh and Nakagami-m fading envelopes in which the CDFs of the sum of channel gains, which correspond in this case to the CDFs of the sum of independent exponentials and Gamma RVs respectively, are known in closed-from expressions \cite{Botev:2013:SNR:2466677.2466683}\cite{6292935}. A similar observation can be deduced from the i.i.d $\kappa-\mu$ and $\eta-\mu$ fading channels since the sum of i.i.d squared $\kappa-\mu$ and $\eta-\mu$ is again a squared $\kappa-\mu$ and a squared $\eta-\mu$, respectively \cite{4231253}. Moreover, for the correlated Rayleigh fading channels, the CDF of the sum of correlated exponential RVs can be obtained explicitly \cite{1201072}. A further interesting example is for GSC/EGC receivers under i.n.i.d Rayleigh fading channels in which the CDF of the partial sum of ordered i.n.i.d exponential RVs can be shown to admit a closed-form expression \cite{1033001}. Therefore, the set $S$ is chosen as follows
\begin{align}\label{set}
S=\{(r_1,\cdots,r_L), \sum_{i=1}^{L}{r_i^2} \leq \gamma_0^2,r_i\geq 0\},
\end{align} 
and thus $\tilde{P}_{out}$ is the OP at the output of MRC receivers which is given by
\begin{align}
\tilde{P}_{out}=P \left (\sum_{i=1}^{L}{R_i^2} \leq \gamma_0^2 \right ).
\end{align}
In other words, based on the knowledge of a closed-form expression of the OP at the output of MRC receivers, we construct an IS estimator of OP values at the output of EGC diversity receivers. In the next section, we provide more details on the implementation of the above IS scheme for the case of i.n.i.d Rayleigh, correlated Rayleigh and i.i.d Rice fading channels. Furthermore, we extend our approach to the case of GSC/EGC receivers under the i.n.i.d Rayleigh fading channels. We perform for each case a theoretical study of the proposed estimator and show that it achieves the bounded relative error property. We note here that the considered scenarios are illustrations of our approach that can be applicable to other scenarios such as Nakagami-m, $\kappa-\mu$, and $\eta-\mu$ fading channels. 

The squared coefficient of variation, defined as the ratio between the variance of an estimator to its squared mean, of the proposed IS estimator is given by
\begin{align}
\frac{\mathrm{var}_g \left [ \bold{1}_{\left (\sum_{i=1}^{L}{R_i} \leq \gamma_0\right )} L(R_1,\cdots,R_L)\right ]}{P_{out}^2}=\frac{\tilde{P}_{out}}{P_{out}}-1.
\end{align}
Therefore, the closer $\tilde{P}_{out}$ is to $P_{out}$, the smaller the coefficient of variation is, and hence the more efficient the proposed estimator is. Particularly, the bounded relative error holds when $\tilde{P}_{out}/P_{out}$ is bounded for a sufficiently small threshold. 
\subsection{Independent Rayleigh Fading Channels}
We consider the first case study where $R_i$, $i=1,2,\cdots,L$, have i.n.i.d Rayleigh distributions. Hence, the PDF $f(\cdot)$ is given by 
\begin{align}
f(r_1,\cdots,r_L)=\prod_{i=1}^{L}{f_{R_i}(r_i)},
\end{align}
where the univariate PDF of $R_i$ is given by
\begin{align}\label{rayleigh}
f_{R_i}(r)=\frac{2r}{\Omega_i} \exp \left (-r^2/\Omega_i \right ), \hspace{4mm} r \geq 0.
\end{align}
Next, in order to apply our proposed IS approach, it is essential to provide a closed-form expression of the quantity $\tilde{P}_{out}$. This expression is obtained from \cite{Botev:2013:SNR:2466677.2466683,8472928} as follows
\begin{align}
\tilde{P}_{out}=1-(1,0,\cdots,0)\exp \left ( \gamma_0^2 \bold{A(\Omega)} \right ) (1,1,\cdots,1)',
\end{align}
with $\bold{\Omega}=(\Omega_1,\cdots,\Omega_L)^T$, $\exp \left ( \gamma_0^2 \bold{A(\Omega)} \right )$ denotes the matrix exponential of $\gamma_0^2 \bold{A(\Omega)}$ and 
\begin{align}
\bold{A(\Omega)} = 
 \begin{pmatrix}
  -1/\Omega_1 & 1/\Omega_1     & 0    & \cdots & 0 \\
    0       & -1/\Omega_2   & 1/\Omega_2 & \cdots & 0 \\
  \vdots  & \vdots  & \ddots & \ddots &\vdots  \\
    0 & \cdots & 0 & -1/\Omega_{N-1} &1/\Omega_{N-1} \\
  0 & \cdots & 0 & 0 &-1/\Omega_N 
 \end{pmatrix}
\end{align}
In the implementation of the proposed IS estimator, one has to be able to efficiently sample from the biased PDF $g(\cdot)$ given in (\ref{biased_pdf}), that is, the truncation of the underlying PDF $f(\cdot)$ over the set $S$ given in (\ref{set}). To do that, we denote by $G_i=R_i^2/\gamma_0^2$, $i=1,2,\cdots,L$, and thus our problem reduces to sampling $G_1,\cdots, G_L$ according to their underlying PDF truncated over the set $\{\sum_{i=1}^{L}{G_i} \leq 1\}$. To this end, we propose to use the acceptance-rejection technique with proposal PDF the uniform distribution over the unit simplex $\{\sum_{i=1}^{L}{G_i} \leq 1\}$. The whole procedure is described in Algorithm 1.
\begin{algorithm}[H]
\caption{Samples for the independent Rayleigh case}
\begin{algorithmic}[1]\label{Algo1}
\STATE \textbf{Inputs:} $\{\Omega_i \}_{i=1}^{L}$ and $\gamma_0$.
\STATE \textbf{Outputs:} $\{R_i \}_{i=1}^{L}$.
\WHILE {$U > \exp \left (-\gamma_0^2\sum_{i=1}^{N}{U_i/\Omega_i} \right )$} 
\STATE Generate $\{U_i\}_{i=1}^{N}$  from the uniform distribution over the set $\{u_i \geq 0, \sum_{i=1}^{N}{u_i}\leq 1\}$, see \cite[Algorithm 3.23]{opac-b1132466}.
\STATE Generate a sample $U$ from the uniform distribution over $[0,1]$.
\ENDWHILE
\STATE $\bold{G} \leftarrow \bold{U}$.
\STATE Set $R_i \leftarrow \gamma_0\sqrt{ G_i}$.
\end{algorithmic}
\end{algorithm}
We now provide a theoretical efficiency result of the proposed IS estimator. In fact, we show in the following proposition that it has a bounded relative error.
\begin{prop}
\hspace{2mm} In the case of independent Rayleigh fading channels, the proposed IS estimator of $P_{out}$ achieves the bounded relative error property, that is
\begin{align}
\limsup_{\gamma_0\rightarrow 0}{\frac{\tilde{P}_{out}}{P_{out}}} < \infty.
\end{align}
\end{prop}
\begin{proof}
We first upper bound the quantity $\tilde{P}_{out}$ as follows
\begin{align}
\nonumber \tilde{P}_{out}&= P \left ( \sum_{i=1}^{L}{R_i^2} \leq \gamma_0^2\right ) \\
\nonumber & \leq P \left (R_1 \leq \gamma_0, \cdots, R_L \leq \gamma_0 \right )\\
&=\prod_{i=1}^{L}{ \left (1-\exp \left (-\gamma_0^2/\Omega_i \right ) \right )}.
\end{align}
Then, we lower bound $P_{out}$
\begin{align}
P_{out} &= P \left (\sum_{i=1}^{L}{R_i}  \leq \gamma_0\right )\\
\nonumber & \geq P \left ( R_1 \leq \gamma_0/L, \cdots, R_L \leq \gamma_0/L\right )\\
&= \prod_{i=1}^{L}{ \left (1-\exp \left ( -\gamma_0^2/L^2\Omega_i\right ) \right )}.
\end{align}
Therefore, we obtain the following result
\begin{align}
\frac{\tilde{P}_{out}}{P_{out}} \leq  \frac{\prod_{i=1}^{L}{ \left (1-\exp \left (-\gamma_0^2/\Omega_i \right ) \right )}}{\prod_{i=1}^{L}{ \left (1-\exp \left ( -\gamma_0^2/L^2\Omega_i\right ) \right )}}.
\end{align}
Applying the limit superior on both side, it follows
\begin{align}
\limsup_{\gamma_0 \rightarrow 0} {\frac{\tilde{P}_{out}}{P_{out}}} \leq L^{2L},
\end{align}
and hence the proof is concluded.
\end{proof}
\subsection{Correlated Rayleigh Fading Channels}
Here we consider the case where the Rayleigh fading channels are correlated. The correlation model that we adopt is presented in \cite{1201072}, where the correlated Rayleigh RVs are generated from the correlated Gaussian RVs.  More specifically, we consider two $L$ dimensional Gaussian random vectors $\bold{X}$ and $\bold{Y}$ with zero means and same covariance matrices $\bold{\Sigma}$. We assume for simplicity that $\mathbb{E} \left [ \bold{X} \bold{Y}^{T}\right ]=0$ (the cross covariance matrix is zero). We define the random vector $\bold{R}$ as follows
\begin{align}
R_i=\sqrt{X_i^2+Y_i^2}, \hspace{2mm} i=1,\cdots,L.
\end{align}
Thus, we can see that $\bold{R}$ is a multivariate Rayleigh random vector with correlated components. We settle for a particular structure of the covariance matrix $\bold{\Sigma}$. In fact, we assume that $\bold{\Sigma}$ is a matrix of exponential correlations, that is
\begin{align}
\Sigma_{ij} = \begin{cases} \sigma^2, &\mbox{if } i = j \\
\rho^{|i-j|}\sigma^2, & \mbox{if } i\neq j \end{cases}. 
\end{align}
With this structure of the covariance matrix, the multivariate Rayleigh PDF is given by \cite{1201072}
\begin{align}
\nonumber &f(r_1,\cdots,r_L)=\frac{\prod_{i=1}^{L}{r_i}}{\sigma^{2L}(1-\rho^2)^{L-1}} \\
\nonumber & \times \exp \left ( -\frac{1}{2(1-\rho^2)\sigma^2} \left [r_1^2+r_L^2+(1+\rho^2)\sum_{i=2}^{L-1}{r_i^2} \right ]\right )\\
&\times \prod_{i=1}^{L-1}{I_0 \left (\frac{\rho}{(1-\rho^2)\sigma^2}r_ir_{i+1} \right )}, \hspace{4mm}r_1,r_2,\cdots,r_L \geq 0,
\end{align}
where $I_0(\cdot)$ denotes the zero order modified Bessel function of the first kind \cite{gradshteyn2007}. Now, we aim to obtain a closed-form expression of $\tilde{P}_{out}$. In our settings, it was proven in \cite[Eq.104]{1201072} that the moment generating function of $\sum_{i=1}^{L}{R_i^2}$ is given by
\begin{align}
M_{\sum_{i=1}^{L}{R_i^2}}(s)=\frac{1}{\prod_{i=1}^{L}{(1-2s\lambda_i)}},  \hspace{4mm} s <\frac{1}{2 \lambda_i} \text{  for all  }i
\end{align}
where $\lambda_i$, $i=1,\cdots,L$, are the eigenvalues of the Gaussian covariance matrix $\Sigma$.  Therefore, we deduce that $\sum_{i=1}^{L}{R_i^2}$ has the same distribution as the sum of $L$ independent exponential RVs with means $2\lambda_i$, $i=1,2,\cdots,L$. Hence,  the quantity $\tilde{P}_{out}$ is expressed as
\begin{align}
\tilde{P}_{out}=1-(1,0,\cdots,0)\exp \left ( \gamma_0^2 \bold{A(\boldsymbol{2\lambda})} \right ) (1,1,\cdots,1)',
\end{align}
with $\boldsymbol{\lambda}=(\lambda_1,\cdots,\lambda_L)^T$. The remaining step is then to provide an algorithm in order to sample from the biased PDF $g(\cdot)$. To do that, we proceed as in the previous example by applying the acceptance-rejection technique with a uniform distribution over the unit simplex $\{\sum_{i=1}^{L}{G_i} \leq 1 \}$ as a proposal. The following algorithm provides the necessary details to perform the sampling.
\begin{algorithm}[H]
\caption{Samples for the correlated Rayleigh case}
\begin{algorithmic}[1]\label{Algo1}
\STATE \textbf{Inputs:} $\sigma$, $\rho$ and $\gamma_0$.
\STATE \textbf{Outputs:} $\{R_i \}_{i=1}^{L}$.
\WHILE {$U > \exp \left (-\frac{\gamma_0^2 \left [ U_1+U_L+(1+\rho^2)\sum_{i=2}^{L-1}{U_i}\right ]}{2(1-\rho^2)\sigma^2}  \right ) 
\newline \times \prod_{i=1}^{L-1}{\frac{I_0 \left (\frac{\rho \gamma_0^2 \sqrt{U_iU_{i+1}}}{(1-\rho^2)\sigma^2} \right )}{I_0 \left (\frac{\rho \gamma_0^2}{(1-\rho^2)\sigma^2} \right )}}$}
\STATE Generate $\{U_i\}_{i=1}^{N}$  from the uniform distribution over the set $\{u_i \geq 0, \sum_{i=1}^{N}{u_i}\leq 1\}$.
\STATE Generate a sample $U$ from the uniform distribution over $[0,1]$.
\ENDWHILE
\STATE $\bold{G} \leftarrow \bold{U}$.
\STATE Set $R_i \leftarrow \gamma_0\sqrt{ G_i}$.
\end{algorithmic}
\end{algorithm}
Next, we study the efficiency of the proposed estimator and investigate whether the bounded relative error property holds for this scenario as well. 
\begin{prop}
\hspace{2mm} In the case of correlated Rayleigh fading channels, the proposed IS estimator of $P_{out}$ achieves the bounded relative error property
\begin{align}
\limsup_{\gamma_0 \rightarrow 0} \frac{\tilde{P}_{out}}{P_{out}} <\infty 
\end{align}
\end{prop}
\begin{proof}
We follow the same steps as in the proof of Proposition 1. In fact, we have 
\begin{align}\label{ineq1}
\frac{\tilde{P}{out}}{P_{out}} \leq \frac{P \left (R_1 \leq \gamma_0, \cdots, R_L \leq \gamma_0 \right )}{P \left (R_1 \leq \gamma_0/L,\cdots, R_L \leq \gamma_0/L \right )}.
\end{align}
Then, we use the following asymptotic result of the multivariate CDF of the Rayleigh random vector which is given in \cite{1221771}
\begin{align}
P \left (R_1 \leq \gamma_0,\cdots, R_L \leq \gamma_0 \right ) \sim a \gamma_0^{2L} \text{,   as  }\gamma_0 \rightarrow 0.
\end{align}
This result concludes the proof. 
\end{proof}
\subsection{i.i,d Rician Fading Channels}
Here, we explore the case where the $R_i$, $i=1,\cdots,L$, are i.i.d Rician fading channels with a common PDF 
\begin{align}
\nonumber f_{R_i}(r)&= \frac{2r(K+1)}{\Omega}\exp \left (-K-\frac{K+1}{\Omega}r^2 \right )\\
& \times I_0 \left (2r \sqrt{\frac{K(K+1)}{\Omega}} \right ), \hspace{2mm} r \geq 0,
\end{align}
where $K$ is the Rice factor and $\Omega=\mathbb{E} \left [R_i^2 \right ]$, for all $i\in \{1,2,\cdots, L \}$. 

In order to obtain an expression of $\tilde{P}_{out}$, we use the fact that the sum of i.i.d squared Rician (equivalently the sum of i.i.d non centered Chi squared RVs) is a squared $\kappa-\mu$ RV with parameters $\kappa=K$ and $\mu=L$ and average power equal to $\tilde{\Omega}=L \Omega$ \cite{4231253,4570452}. More precisely, the PDF of $\sum_{i=1}^{L}{R_i^2}$ is given by
\begin{align}
\nonumber f_{\sum_{i=1}^{L}{R_i^2}}(r)&=\frac{L(1+K)^{\frac{L+1}{2}} r^{\frac{L-1}{2}}}{\tilde{\Omega}^{\frac{L+1}{2}} K^{\frac{L-1}{2}} \exp(LK)} \exp \left ( -\frac{(1+K)Lr}{\tilde{\Omega}}\right )\\
&\times I_{L-1} \left (2L \sqrt{\frac{K(K+1)r}{\tilde{\Omega}}} \right ), \hspace{2mm} r\geq 0.
\end{align}
Therefore, the quantity $\tilde{P}_{out}$ is expressed as
\begin{align}
\tilde{P}_{out}=1-Q_L \left (\sqrt{2KL},\sqrt{\frac{2(K+1)L}{\tilde{\Omega}}} \gamma_0 \right ),
\end{align}
where $Q_{\mu}(\cdot,\cdot)$ is the generalized Marcum $Q$ function \cite{andras2011generalized}.  

Similarly to the previous cases, sampling according to the biased PDF $g(\cdot)$ is easily performed using the acceptance-rejection approach. 
\begin{algorithm}[H]
\caption{Samples for the i.i.d Rice case}
\begin{algorithmic}[1]\label{Algo1}
\STATE \textbf{Inputs:} $K$, $\Omega$ and $\gamma_0$.
\STATE \textbf{Outputs:} $\{R_i \}_{i=1}^{L}$.
\WHILE {$U > \exp \left ( -\frac{(K+1)}{\Omega} \gamma_0^2 \sum_{i=1}^{L}{G_i}\right )
\newline \prod_{i=1}^{L}{\frac{I_0 \left (2 \sqrt{\frac{K(K+1)\gamma_0^2 G_i}{\Omega}} \right )}{I_0 \left ( 2\sqrt{\frac{K(K+1)\gamma_0^2}{\Omega}}\right )}}$}
\STATE Generate $\{U_i\}_{i=1}^{N}$  from the uniform distribution over the set $\{u_i \geq 0, \sum_{i=1}^{N}{u_i}\leq 1\}$.
\STATE Generate a sample $U$ from the uniform distribution over $[0,1]$.
\ENDWHILE
\STATE $\bold{G} \leftarrow \bold{U}$.
\STATE Set $R_i \leftarrow \gamma_0\sqrt{ G_i}$.
\end{algorithmic}
\end{algorithm}
Next we show that the bounded relative error holds again for the case of i.i.d Rician fading channels. 
\begin{prop}
\hspace{2mm} In the case of i.i.d Rice fading channels, the proposed IS estimator of $P_{out}$ achieves the bounded relative error property
\begin{align}
\limsup_{\gamma_0 \rightarrow 0} {\frac{\tilde{P}_{out}}{P_{out}}} < \infty.
\end{align}
\end{prop}
\begin{proof} 
First, the CDF of the Rice fading envelope is given by
\begin{align}
P(R_i \leq \gamma_0)=1-Q_1(\sqrt{2K},\sqrt{\frac{2(K+1)}{\Omega}}\gamma_0).
\end{align}
Then, the proof is based on the following asymptotic which is obtained from \cite{andras2011generalized,7769235}
\begin{align}\label{asym}
P(R_i \leq \gamma_0) \sim \frac{(K+1)\exp \left ( -K\right ) }{\Omega} \gamma_0^2,\hspace{2mm} \gamma_0 \rightarrow 0.
\end{align}
In fact, similarly to the previous proofs, we have from (\ref{ineq1}) that
\begin{align}
\frac{\tilde{P}_{out}}{P_{out}} \leq \frac{\left (P \left (R_1 \leq \gamma_0 \right ) \right )^L}{\left ( P\left ( R_1 \leq \gamma_0/L\right ) \right )^L}.
\end{align}
Using the asymptotic expression in (\ref{asym}), it follows
\begin{align}
\limsup_{\gamma_0 \rightarrow 0} {\frac{\tilde{P}_{out}}{P_{out}}} \leq L^{2L},
\end{align}
and hence the proof is concluded. 
\end{proof}

\subsection{i.n.i.d Ordered Rayleigh RVs}
The fading channel amplitudes $R_i$, $i=1,\cdots,L$ are i.n.i.d Rayleigh with PDF given in (\ref{rayleigh}). In this section, we aim to efficiently estimate OP values when GSC is combined with EGC
\begin{align}\label{pout_ordered}
P_{out}=P \left ( \sum_{i=1}^{N}{R^{(i)}} \leq \gamma_0\right ),
\end{align}
where $N$ satisfies $1 \leq N \leq L$ and denotes the number of selected branches, and $R^{(i)}$ denotes the $i^{th}$ order statistic such that $R^{(1)} \geq R^{(2)}\geq \cdots \geq R^{(L)}$. Note that $\gamma_0$ is given in this case by $\sqrt{\gamma_{th} N N_0/E_s}$. There are few existing works that have computed the above probability when the RVs $R_i$, $i=1,\cdots,L$ are either exponentials or Gamma distributed \cite{5605378,7953495}. These results can help to compute OP values at the output of GSC/MRC receivers.
When GSC is combined with EGC, a competitor of the present work is in \cite{8472928} where the authors proposed two variance reduction techniques based on IS and conditional MC (another type of variance reduction technique). However, the conditional MC estimator described in \cite{8472928} is only applicable when the Rayleigh RVs are i.i.d. Moreover, the construction of the IS estimator in \cite{8472928} is based on a choice of $S$ given by $S=\{(r_1,\cdots,r_L), \max_{1\leq i \leq L}{r_i} \leq \gamma_0, r_i \geq 0\}$. Therefore, given that this choice contains our choice of  $S$ in (\ref{set}), we conclude that our proposed estimator  is more efficient than the IS estimator proposed in \cite{8472928}. We verify this conclusion in the numerical results section.

We now show how we can use our proposed IS approach for the present case as well. Let $h_i=R_i^2$, $i=1,\cdots L$,  be the channel gains which are i.n.i.d exponential RVs with means $\Omega_i$. Then, the quantity $\tilde{P}_{out}$ is given by the partial sum of the ordered exponential RVs
\begin{align}
\tilde{P}_{out}=P \left ( \sum_{i=1}^{N}{h^{(i)}} \leq \gamma_0^2\right ).
\end{align}
In order to compute $\tilde{P}_{out}$, we introduce the following RVs
\begin{align}\label{rep_order_inid}
X_i=h^{(i)}-h^{(i+1)}, \text{  } i=1,2,\cdots,L-1, \text{ } X_L=h^{(L)}
\end{align}
Thus, with this representation, we get
\begin{align}
\tilde{P}_{out}=P \left ( \sum_{i=1}^{L}{\alpha_i X_i} \leq \gamma_0^2\right ),
\end{align}
with 
\begin{align}
\alpha_i=
\begin{cases}
i, \text{   } i=1,\cdots,N\\
N, \text{   } i=N+1,\cdots,L.
\end{cases}
\end{align}
Moreover, it was shown in \cite{1033001} that the joint PDF of $\bold{X}=(X_1,\cdots,X_L)^t$  is  given as follows
\begin{align}
f_{\bold{X}}(x_1,\cdots,x_L)=\sum_{\substack{i_1,i_2,\cdots,i_L=1 \\ i_1\neq i_2\neq \cdots \neq i_L}}^{L}{\prod_{\ell=1}^{L}{\frac{1}{\Omega_{i_{\ell}}}} \exp \left (-x_{\ell} \sum_{k=1}^{\ell}{\frac{1}{\Omega_{i_{k}}}}\right )}.
\end{align}
Interestingly, we observe that while the components of $\bold{X}$ are dependent, their joint PDF is given by the sum of products of independent exponentials. Therefore, by using the formula of the CDF of the sum of independent exponentials, we easily obtain a closed-form expression of $\tilde{P}_{out}$: 
\begin{align}
\tilde{P}_{out}=\prod_{\ell=1}^{L}{\frac{1}{\Omega_{\ell}}}\sum_{\substack{i_1,i_2,\cdots,i_L=1 \\ i_1\neq i_2\neq \cdots \neq i_L}}^{L}{\left (\prod_{\ell=1}^{L}{\frac{1}{\sum_{k=1}^{\ell}{\frac{1} {\Omega_{i_k}}}}} \right )} \tilde{P}_{out,i_1,\cdots,i_L}
\end{align}
 $ \tilde{P}_{out,i_1,\cdots,i_L}=1- (1,0,\cdots,0) \exp \left ( \gamma_0^2 A(\boldsymbol{\tilde{\alpha}})\right ) (1,\cdots,1)^t$ and  $\tilde{\alpha}_{i}=\frac{\alpha_{i}}{\sum_{k=1}^{i}{\frac{1}{\Omega_{i_k}}}}$, $i=1,2,\cdots,L$.


Next, we show how sampling according to the biased PDF is performed. We exploit the representation (\ref{rep_order_inid}) and sample from $X_1, \cdots, X_L$ truncated over $\{ \sum_{i=1}^{L}{\alpha_i X_i} \leq \gamma_0^2 \}$. By letting $G_i=\alpha_i X_i/\gamma_0^2$, $i=1,\cdots,L$, we construct the following algorithm.

\begin{algorithm}[H]
\caption{Samples for the independent ordered Rayleigh}
\begin{algorithmic}[1]\label{Algo1}
\STATE \textbf{Inputs:} $\gamma_0$, and $\{\Omega_i\}_{i=1}^{L}$.
\STATE \textbf{Outputs:} $\{h^{(i)}\}_{i=1}^{N}$.
\STATE Sample a permutation $(i_1,\cdots,i_L)$ from the discrete distribution with probability $p(i_1, \cdots,i_L)=\frac{\tilde{P}_{out,i_1,\cdots,i_L}}{\tilde{P}_{out}}\prod_{\ell=1}^{L}{\frac{1}{\Omega_{i_{\ell}}\sum_{k=1}^{\ell}{\frac{1}{\Omega_{i_k}}}}}$
\WHILE {$U > \exp \left ( - \gamma_0^2\sum_{\ell=1}^{L}{\frac{U_{\ell}}{\alpha_{\ell}} \sum_{k=1}^{\ell}{\frac{1}{\Omega_{i_k}}}  }\right )$} 
\STATE Generate $\{U_i\}_{i=1}^{N}$  from the uniform distribution over the set $\{u_i \geq 0, \sum_{i=1}^{N}{u_i}\leq 1\}$.
\STATE Generate a sample $U$ from the uniform distribution over $[0,1]$.
\ENDWHILE
\STATE $\bold{G} \leftarrow \bold{U}$.
\STATE Set $X_i \leftarrow \gamma_0^2\ G_i/\alpha_i$.
\STATE Compute $\{ h^{(i)}\}_{i=1}^{N}$ from (\ref{rep_order_inid}).
\end{algorithmic}
\end{algorithm}
In this case, we can also show that the bounded relative error property holds.
\begin{prop}
\hspace{2mm} In the case of i.n.i.d  Rayleigh fading channels at the output of GSC/EGC receivers, the proposed IS estimator of $P_{out}$ achieves the bounded relative error property
\begin{align}
\limsup_{\gamma_0 \rightarrow 0} {\frac{\tilde{P}_{out}}{P_{out}}} < \infty.
\end{align}
\end{prop}
\begin{proof}
First, we upper bound $\tilde{P}_{out}$ as follows
\begin{align}
\nonumber \tilde{P}_{out}&=P \left ( \sum_{i=1}^{N}{h^{(i)}} \leq \gamma_0^2\right )\\
& \leq P \left ( h^{(1)} \leq \gamma_0^2\right )= \prod_{i=1}^{L}{\left ( 1- \exp \left ( -\gamma_0^2/\Omega_i\right )\right )}.
\end{align}
On the other hand, we have
\begin{align}
\nonumber P_{out}=&P \left ( \sum_{i=1}^{N}{R^{(i)}} \leq \gamma_0\right )\\
\nonumber & \geq P \left (R^{(1)} \leq \gamma_0/N,\cdots, R^{(N)} \leq \gamma_0/N\right )\\
&=\prod_{i=1}^{L}{\left ( 1-\exp \left (-\frac{\gamma_0^2}{N^2\Omega_i} \right )\right )}.
\end{align}
Thus, we obtain
\begin{align}
\frac{\tilde{P}_{out}}{P_{out}} \leq N^{2L},
\end{align}
and hence the proof is concluded.
\end{proof}

\section{Simulation Results}
In this section, we present some simulations to illustrate our theoretical results. Furthermore, we study the efficiency of the proposed estimator with respect to other estimators including the naive MC one. Before showing the results, we define a performance metric that will serve as a measure of efficiency of an estimator. We define the relative error of the naive MC estimator as the relative half-width of its confidence interval 
\begin{align}\label{errmc}
\epsilon_{MC}=\frac{C \sqrt{P_{out}(1-P_{out})}}{P_{out} \sqrt{M}},
\end{align}
where $C$ is the confidence constant chosen to be equal to $1.96$ (corresponding to $ 95\%$ confidence level). The relative error of the proposed estimator is given using a similar argument by
\begin{align}\label{erris}
\epsilon_{IS}=\frac{C \sqrt{\frac{\tilde{P}_{out}}{P_{out}}-1}}{\sqrt{M}}.
\end{align}
We performed the comparison between different estimators in terms of the necessary number of simulation runs in order to meet a fixed accuracy requirement measured by the above quantities. More specifically, we set $\epsilon_{MC}$ and $\epsilon_{IS}$ equal to a fixed value and use (\ref{errmc}) and (\ref{erris}) to find the number of simulation runs needed to meet this fixed accuracy requirement. 

In the first experiment, we consider the i.i.d Rayleigh fading channels and we evaluate the OP under EGC using the proposed estimator as well the second estimator of \cite{7328688}, which is based on the use of the hazard rate twisting (HRT) technique. Then we investigate the efficiency of both estimators using the number of simulation runs required to meet a fixed accuracy level. The same steps are repeated for two other experiments; the correlated Rayleigh with exponential correlation and the ordered i.n.i.d Rayleigh scenarios. In the former experiment, we make the comparison with respect to the naive MC estimator since, to the best of our knowledge, this problem has not been investigated by existing estimators. In the latter case, i.e., in the ordered i.n.i.d Rayleigh case, we perform the comparison with the universal IS estimator of \cite{8472928}, as well as with the naive MC estimator.
\subsection{i.i.d Rayleigh Fading Channels}
In Fig. \ref{fig1}, we plot the estimated value of $P_{out}$ given by naive MC simulations, the HRT method and the proposed estimator for the case of i.i.d Rayleigh fading channels. The plot is a function of the threshold value $\gamma_{th}$ and for three different values of the number of diversity branches $L$. 

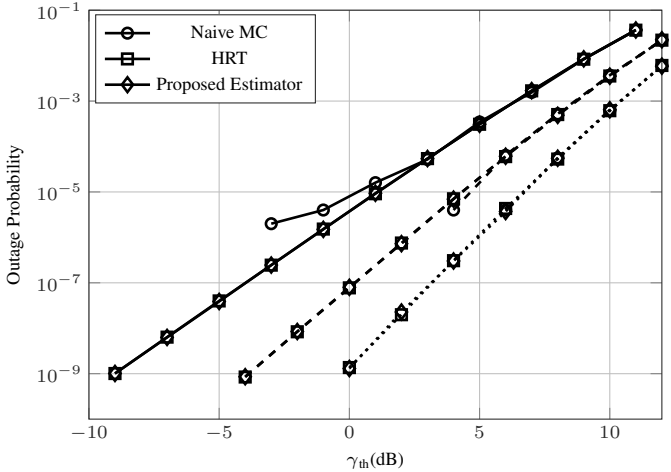
\begin{figure}[h]
\centering
\setlength\figureheight{0.30\textwidth}
\setlength\figurewidth{0.42\textwidth}
%
%
%
%
\begin{tikzpicture}
\scalefont{0.7}
\begin{semilogyaxis}[%
width=\figurewidth,
height=\figureheight,
scale only axis,
every outer x axis line/.append style={darkgray!60!black},
every x tick label/.append style={font=\color{darkgray!60!black}},
xmin=-10, xmax=12,
xminorticks=true,
xlabel={$\gamma{}_{\text{th}}\text{(dB)}$},
xmajorgrids,
xminorgrids,
every outer y axis line/.append style={darkgray!60!black},
every y tick label/.append style={font=\color{darkgray!60!black}},
ymin=1e-10, ymax=0.1,
yminorticks=true,
ylabel={Outage Probability},
ymajorgrids,
yminorgrids,
grid style={solid},
legend style={at={(0.010856448509367,0.77188527036781)},anchor=south west,draw=darkgray!60!black,fill=white,align=left}]

\addplot [
color=black,
solid,
line width=1.0pt,
mark size=2.0pt,
mark=o,
mark options={solid},
]
coordinates{
 (-5,0)(-3,2e-06)(-1,4e-06)(1,1.6e-05)(3,5.2e-05)(5,0.000346)(7,0.001552)(9,0.00835)(11,0.036504) 
};
\addlegendentry{Naive MC};

\addplot [
color=black,
solid,
line width=1.0pt,
mark size=2.0pt,
mark=square,
mark options={solid},
]
coordinates{
 (-9,1.00852887607274e-09)(-7,6.37654131743731e-09)(-5,3.99532769770398e-08)(-3,2.44815873022387e-07)(-1,1.5244722526754e-06)(1,9.13395515742481e-06)(3,5.41909331030652e-05)(5,0.00030874862645797)(7,0.00167941286768508)(9,0.00834886892943891)(11,0.0363302112922912) 
};
\addlegendentry{HRT};

\addplot [
color=black,
solid,
line width=1.0pt,
mark size=2.8pt,
mark=diamond,
mark options={solid},
]
coordinates{
 (-9,1.0127215934268e-09)(-7,6.28947345928332e-09)(-5,3.88720313392177e-08)(-3,2.45849072375467e-07)(-1,1.50113328508925e-06)(1,9.35607738783062e-06)(3,5.43711952911368e-05)(5,0.000308709905836493)(7,0.00169711179723318)(9,0.00849460233392041)(11,0.0367268885628281) 
};
\addlegendentry{Proposed Estimator};

\addplot [
color=black,
dashed,
line width=1.0pt,
mark size=2.0pt,
mark=o,
mark options={solid},
forget plot
]
coordinates{
 (2,0)(4,4e-06)(6,6.4e-05)(8,0.000522)(10,0.0037)(12,0.021736) 
};
\addplot [
color=black,
dashed,
line width=1.0pt,
mark size=2.0pt,
mark=square,
mark options={solid},
forget plot
]
coordinates{
 (-4,8.45333775116754e-10)(-2,8.36308830188868e-09)(0,7.74890211013281e-08)(2,7.51164550077172e-07)(4,7.0636745420444e-06)(6,5.99929391399287e-05)(8,0.000505005083998958)(10,0.00356270891983513)(12,0.0219019418898135) 
};
\addplot [
color=black,
dashed,
line width=1.0pt,
mark size=2.8pt,
mark=diamond,
mark options={solid},
forget plot
]
coordinates{
 (-4,8.57853720199753e-10)(-2,8.48030648256826e-09)(0,7.96279141569371e-08)(2,7.27251937228643e-07)(4,6.94353001698618e-06)(6,6.10428491458308e-05)(8,0.000484315552387385)(10,0.00352914317464427)(12,0.0223496671570969) 
};
\addplot [
color=black,
dotted,
line width=1.0pt,
mark size=2.0pt,
mark=o,
mark options={solid},
forget plot
]
coordinates{
 (4,0)(6,4e-06)(8,5.6e-05)(10,0.000618)(12,0.005998) 
};
\addplot [
color=black,
dotted,
line width=1.0pt,
mark size=2.0pt,
mark=square,
mark options={solid},
forget plot
]
coordinates{
 (0,1.36561966570683e-09)(2,1.98572797111623e-08)(4,3.07211786773756e-07)(6,4.32899832546696e-06)(8,5.29236557871607e-05)(10,0.000614256232635961)(12,0.00607579188774949) 
};
\addplot [
color=black,
dotted,
line width=1.0pt,
mark size=2.8pt,
mark=diamond,
mark options={solid},
forget plot
]
coordinates{
 (0,1.3006645098157e-09)(2,2.27888749107663e-08)(4,3.02470946657774e-07)(6,3.70527027243265e-06)(8,5.51878674670409e-05)(10,0.000644281788641576)(12,0.00586894798709869) 
};
\end{semilogyaxis}
\end{tikzpicture}%
\caption{Outage Probability for $L=4,5,6$ branch EGC receiver with i.i.d Rayleigh fading channels as a function of $\gamma_{th}$. $L=4$ (solid line), $L=5$ (dashed line), and $L=6$ (dotted line). The system parameters are $E_s/N_0=1$ dB, $\Omega=10$ dB, and $M=5 \times 10^5$.}
\label{fig1}
\end{figure}

This figure reveals the failure of naive MC simulations. In fact, the naive estimator loses its accuracy when the value of $P_{out}$ decreases, i.e. in the region of rare events. Thus, more than $5\times 10^5$ samples are required in order for the naive sampler to retrieve a good level of accuracy. The opposite observation can be easily deduced regarding the accuracy of the proposed estimator and the HRT method. In fact, using the same number of simulation runs, these two estimators coincide perfectly and yield very accurate estimates of $P_{out}$ in the considered range of OP values. 

We now investigate the efficiency of these estimators in terms of the number of simulation runs needed to meet a fixed accuracy requirement. More precisely, we compute from (\ref{errmc}) and (\ref{erris}) the number of simulation runs needed to ensure that $\epsilon_{MC}=\epsilon_{IS}=\epsilon_{HRT}=5 \%$. Note that $\epsilon_{HRT}$ is given by a similar expression as in  (\ref{errmc}) and (\ref{erris}). In Fig. \ref{fig2}, we plot the number of samples needed by the naive MC simulation, the proposed method, and the HRT technique as a function of $\gamma_{th}$ and for the three values of $L$ as in Fig. \ref{fig1}. 

\begin{figure}[h]
\centering
\setlength\figureheight{0.30\textwidth}
\setlength\figurewidth{0.42\textwidth}
%
%
%
%
\begin{tikzpicture}
\scalefont{0.7}
\begin{semilogyaxis}[%
width=\figurewidth,
height=\figureheight,
scale only axis,
every outer x axis line/.append style={darkgray!60!black},
every x tick label/.append style={font=\color{darkgray!60!black}},
xmin=-10, xmax=12,
xminorticks=true,
xlabel={$\gamma{}_{\text{th}}\text{(dB)}$},
xmajorgrids,
xminorgrids,
every outer y axis line/.append style={darkgray!60!black},
every y tick label/.append style={font=\color{darkgray!60!black}},
ymin=10000, ymax=10000000000000,
yminorticks=true,
ylabel={Number of Simulation Runs},
ymajorgrids,
yminorgrids,
grid style={solid},
legend style={at={(0.620856448509367,0.77188527036781)},anchor=south west,draw=darkgray!60!black,fill=white,align=left}]

\addplot [
color=black,
solid,
line width=1.0pt,
mark size=2.0pt,
mark=o,
mark options={solid},
]
coordinates{
 (-9,1523645018905.16)(-7,240983303283.78)(-5,38460923730.7184)(-3,6276715659.14199)(-1,1007980076.21208)(1,168232264.984359)(3,28354498.4383678)(5,4975457.16116653)(7,913449.799348961)(9,182517.033975125)(11,40759.8384772957) 
};
\addlegendentry{Naive MC};
\addplot [
color=black,
solid,
line width=1.0pt,
mark size=2.0pt,
mark=square,
mark options={solid},
]
coordinates{
 (-9,44307.417373035)(-7,43841.7852814861)(-5,43802.7359630819)(-3,44327.6066612948)(-1,43232.8335235567)(1,43703.7126865678)(3,43350.1314927385)(5,43010.3728103533)(7,41944.0177016656)(9,41079.8827339142)(11,40762.4945980065) 
};
\addlegendentry{HRT};
\addplot [
color=black,
solid,
line width=1.0pt,
mark size=2.8pt,
mark=diamond,
mark options={solid},
]
coordinates{
 (-9,155231.676993088)(-7,154785.134472726)(-5,153398.129388972)(-3,145962.212109808)(-1,140011.171339785)(1,125964.441666048)(3,113653.992104534)(5,94443.5613798777)(7,68886.9624857965)(9,43242.5582700458)(11,22167.4736898439) 
};
\addlegendentry{Proposed Estimator};
\addplot [
color=black,
dashed,
line width=1.0pt,
mark size=2.0pt,
mark=square,
mark options={solid},
forget plot
]
coordinates{
 (-4,180469.493305362)(-2,174681.545959163)(0,182860.200054488)(2,176301.63448687)(4,170696.452048521)(6,175078.035386435)(8,166916.304674265)(10,168962.8291945)(12,168451.296027058) 
};
\addplot [
color=black,
dashed,
line width=1.0pt,
mark size=2.0pt,
mark=o,
mark options={solid},
forget plot
]
coordinates{
 (-4,1817790846566.84)(-2,183740734484.642)(0,19830420607.8161)(2,2045675405.70523)(4,217539629.909168)(6,25612144.2702803)(8,3041284.21208309)(10,429775.609913223)(12,68623.3489055811) 
};
\addplot [
color=black,
dashed,
line width=1.0pt,
mark size=2.8pt,
mark=diamond,
mark options={solid},
forget plot
]
coordinates{
 (-4,1291933.95733484)(-2,1210326.97940021)(0,1141798.97693496)(2,1032543.4869899)(4,800469.136307672)(6,568857.964681157)(8,346277.542821499)(10,157667.985325213)(12,50198.6970822329) 
};
\addplot [
color=black,
dotted,
line width=1.0pt,
mark size=2.0pt,
mark=square,
mark options={solid},
forget plot
]
coordinates{
 (0,856413.757008861)(2,885964.443760553)(4,822392.009013568)(6,795052.560939618)(8,819174.183310132)(10,790302.249234382)(12,757277.924962915) 
};
\addplot [
color=black,
dotted,
line width=1.0pt,
mark size=2.0pt,
mark=o,
mark options={solid},
forget plot
]
coordinates{
 (0,1125232769042.02)(2,77384213338.2813)(4,5001889882.10184)(6,354962795.630617)(8,29033494.6167181)(10,2500090.39828955)(12,251375.24842368) 
};
\addplot [
color=black,
dotted,
line width=1.0pt,
mark size=2.8pt,
mark=diamond,
mark options={solid},
forget plot
]
coordinates{
 (0,12803822.3010524)(2,9145148.31701287)(4,7531027.44162865)(6,5447552.26932609)(8,2355277.88651214)(10,816697.155270141)(12,198599.209139398) 
};
\end{semilogyaxis}
\end{tikzpicture}%
\caption{Number of simulation runs for $L=4,5,6$ branch EGC receiver with i.i.d Rayleigh fading channels as a function of $\gamma_{th}$. $L=4$ (solid line), $L=5$ (dashed line), and $L=6$ (dotted line). The system parameters are $E_s/N_0=1$ dB and $\Omega=10$ dB.}
\label{fig2}
\end{figure}
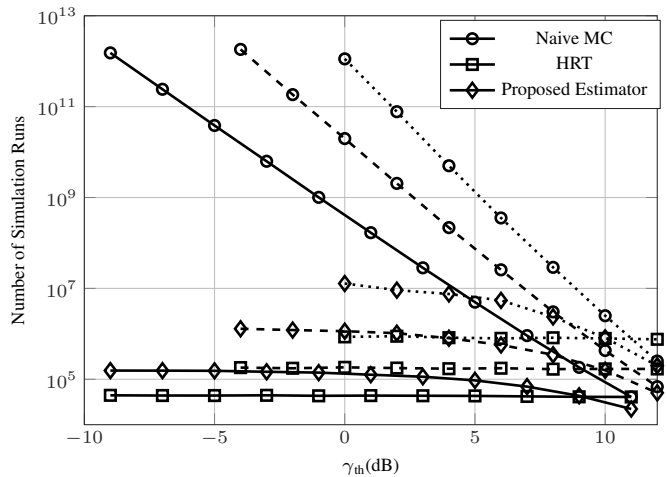

We first observe the high computational effort needed by naive MC simulations in order to achieve a $5\%$ relative error. In fact, the corresponding number of samples is increasing as we decrease the probability of interest $P_{out}$. On the other hand, the computational savings achieved by the proposed IS estimator and the HRT method is obvious and is clearly increasing as we decrease $P_{out}$. More specifically, while the number of samples needed by the naive sampler is increasing as we decrease $\gamma_{th}$, the proposed IS approach and the HRT method require numbers of runs that remain bounded independently of how small $P_{out}$ is. This observation is in accordance with Proposition 1 and the result proven in \cite{7328688} that  show that both estimators have bounded relative errors. For the sake of illustration, for $L=4$ and $\gamma_{th}=-9$ dB, the number of runs needed by naive MC simulation is approximately $1.5 \times 10^{12}$, whereas $1.5 \times 10^5$ and $5 \times 10^4$ samples are required by the proposed approach and the HRT estimator, respectively,  to ensure $5\%$ relative error.      

Note also that the HRT approach performs better than our proposed scheme for the considered values of $L$ and $\gamma_{th}$. Moreover, Fig. \ref{fig2} shows that increasing $L$ has negative effects on the performances of the proposed approach as well as the HRT method. However, this negative effect is more important for the former than the latter. For instance, the HRT approach requires $3.5$ (respectively $15$) times less number of samples than the proposed IS scheme when $P_{out}$ is of the order of $10^{-9}$ and $L=4$ (respectively $L=6$). 

Note however that the outperformance of the HRT approach over our proposed method does not tell the whole story and does not necessarily  exclude our proposed estimator from being a useful technique. In fact, the scope of applicability of our proposed estimator includes the interesting scenario of sums of correlated Rayleigh fading channels with exponential correlation that, to the best of our knowledge, has not been considered by other existing estimators. Moreover, the sum of i.i.d Rice constitutes another argument that shows the relevance of the proposed estimator. In fact, while the HRT estimator is proven to have bounded relative error for the sum of i.i.d Rice variates, it is not clear how sampling according to the HRT biased PDF is performed. On the other hand, we show in Section III-C how our approach can be easily implemented for the i.i.d Rice setting. The same argument holds for the sum of $\kappa-\mu$ RVs as well. Furthermore, our approach is applicable to the case of ordered sum of i.n.i.d Rayleigh RVs which has rarely been investigated. In the following subsections, we apply our proposed estimator to the case of the sum of exponentially correlated Rayleighs and the partial sum of ordered i.n.i.d Rayleighs and determine their computational efficiencies. Note that  we do not include simulations for the i.i.d Rician case  to avoid redundant information and conclusions.
\subsection{Correlated Rayleigh Fading Channels}
Here we consider the case of exponentially correlated Rayleigh fading channels and we aim to perform the same experiment as above. Note that we compare our estimator to only the naive MC method since we are not aware of any other existing estimator for the sum of correlated Rayleigh RVs. In Fig. \ref{fig3}, we plot the estimated value of $P_{out}$ given by the proposed estimator as well as the naive MC method as a function of the threshold and for three different values of $L$.  
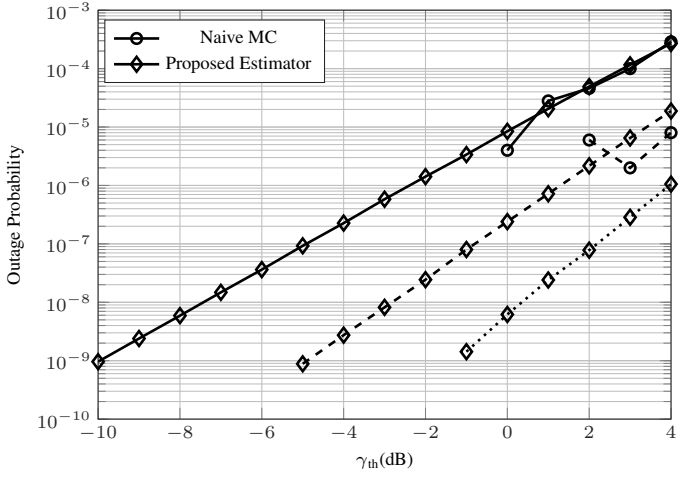
\begin{figure}[h]
\centering
\setlength\figureheight{0.30\textwidth}
\setlength\figurewidth{0.42\textwidth}
%
%
%
%
\begin{tikzpicture}
\scalefont{0.7}
\begin{semilogyaxis}[%
width=\figurewidth,
height=\figureheight,
scale only axis,
every outer x axis line/.append style={darkgray!60!black},
every x tick label/.append style={font=\color{darkgray!60!black}},
xmin=-10, xmax=4,
xminorticks=true,
xlabel={$\gamma{}_{\text{th}}\text{(dB)}$},
xmajorgrids,
xminorgrids,
every outer y axis line/.append style={darkgray!60!black},
every y tick label/.append style={font=\color{darkgray!60!black}},
ymin=1e-10, ymax=0.001,
yminorticks=true,
ylabel={Outage Probability},
ymajorgrids,
yminorgrids,
grid style={solid},
legend style={at={(0.010856448509367,0.82188527036781)},anchor=south west,draw=darkgray!60!black,fill=white,align=left}]
\addplot [
color=black,
solid,
line width=1.0pt,
mark size=2.0pt,
mark=o,
mark options={solid},
]
coordinates{
 (-1,0)(0,4e-06)(1,2.8e-05)(2,4.6e-05)(3,0.0001)(4,0.000288) 
};
\addlegendentry{Naive MC};
\addplot [
color=black,
solid,
line width=1.0pt,
mark size=2.8pt,
mark=diamond,
mark options={solid},
]
coordinates{
 (-10,9.63060136169647e-10)(-9,2.39688677051308e-09)(-8,5.90251098159467e-09)(-7,1.4790177151891e-08)(-6,3.62800915228527e-08)(-5,9.25679415949523e-08)(-4,2.2651829164106e-07)(-3,5.80565574644587e-07)(-2,1.41639756690316e-06)(-1,3.37628043668429e-06)(0,8.44378344863333e-06)(1,2.05460671400883e-05)(2,4.92159345987352e-05)(3,0.000115972555172553)(4,0.000269599117889254) 
};
\addlegendentry{Proposed Estimator};
\addplot [
color=black,
dashed,
line width=1.0pt,
mark size=2.0pt,
mark=o,
mark options={solid},
forget plot
]
coordinates{
 (1,0)(2,6e-06)(3,2e-06)(4,8e-06) 
};
\addplot [
color=black,
dashed,
line width=1.0pt,
mark size=2.8pt,
mark=diamond,
mark options={solid},
forget plot
]
coordinates{
(-5,8.8430e-10)(-4,2.73963084694895e-09)(-3,8.19155766350832e-09)(-2,2.44043764935373e-08)(-1,7.99557000514422e-08)(0,2.40066305879495e-07)(1,7.20122591365792e-07)(2,2.16632004220997e-06)(3,6.49560444340163e-06)(4,1.87958792498669e-05) 
};
\addplot [
color=black,
dotted,
line width=1.0pt,
mark size=2.8pt,
mark=diamond,
mark options={solid},
forget plot
]
coordinates{
 (-1,1.4267e-09)(0,6.1869210563632e-09)(1,2.40655585078144e-08)(2,7.80253887738591e-08)(3,2.85111174509602e-07)(4,1.0553606393579e-06) 
};
\end{semilogyaxis}
\end{tikzpicture}%
\caption{Outage Probability for $L=4,5,6$ branch EGC receiver with exponentially correlated Rayleigh fading channels as a function of $\gamma_{th}$. $L=4$ (solid line), $L=5$ (dashed line), and $L=6$ (dotted line). The system parameters are $E_s/N_0=1$ dB, $\sigma=\sqrt{5}$, $\rho=0.5$,  and $M=5 \times 10^5$.}
\label{fig3}
\end{figure}
The same conclusions can be drawn, as in the previous experiment, on the inability of naive MC simulations using $5\times 10^5$ samples to yield  a precise estimate in the region of small values of $P_{out}$. On the other side, this number of samples is sufficient for our estimator to provide an estimate of $P_{out}$ with a good level of accuracy.  

Next, we quantify the efficiency of the proposed approach with respect to naive MC simulations in terms of necessary number of simulation runs required to ensure a $5\%$ relative error. We plot this number in Fig. \ref{fig4} as a function of $\gamma_{th}$ using the three values of $L$. 
\begin{figure}[h]
\centering
\setlength\figureheight{0.30\textwidth}
\setlength\figurewidth{0.42\textwidth}
%
%
%
%
\begin{tikzpicture}
\scalefont{0.7}
\begin{semilogyaxis}[%
width=\figurewidth,
height=\figureheight,
scale only axis,
every outer x axis line/.append style={darkgray!60!black},
every x tick label/.append style={font=\color{darkgray!60!black}},
xmin=-10, xmax=4,
xminorticks=true,
xlabel={$\gamma{}_{\text{th}}\text{(dB)}$},
xmajorgrids,
xminorgrids,
every outer y axis line/.append style={darkgray!60!black},
every y tick label/.append style={font=\color{darkgray!60!black}},
ymin=10000, ymax=100000000000000,
yminorticks=true,
ylabel={Number of Simulation Runs},
ymajorgrids,
yminorgrids,
grid style={solid},
legend style={at={(0.630856448509367,0.84188527036781)},anchor=south west,draw=darkgray!60!black,fill=white,align=left}]
\addplot [
color=black,
solid,
line width=1.0pt,
mark size=2.0pt,
mark=o,
mark options={solid},
]
coordinates{
 (-10,1595580525876.36)(-9,641098284333.185)(-8,260336659384.717)(-7,103895981873.101)(-6,42354908153.4878)(-5,16600129929.2172)(-4,6783733184.59371)(-3,2646796804.68554)(-2,1084891600.64471)(-1,455126533.677842)(0,181983234.682945)(1,74788445.7699082)(2,31220871.5602801)(3,13248494.780914)(4,5698185.27315263) 
};
\addlegendentry{Naive MC};
\addplot [
color=black,
solid,
line width=1.0pt,
mark size=2.8pt,
mark=diamond,
mark options={solid},
]
coordinates{
 (-10,153086.326666157)(-9,152993.030427176)(-8,154151.611563668)(-7,152127.664255819)(-6,152775.776417213)(-5,146702.125330123)(-4,146018.866174211)(-3,137601.628685083)(-2,134956.7903335)(-1,133850.411842491)(0,124541.718698672)(1,116903.550642947)(2,108917.86190601)(3,100268.251142867)(4,90378.8416315247) 
};
\addlegendentry{Proposed Estimator};
\addplot [
color=black,
dashed,
line width=1.0pt,
mark size=2.0pt,
mark=o,
mark options={solid},
forget plot
]
coordinates{
 (-5,1.7377e12)(-4,560893084373.571)(-3,187588252507.571)(-2,62965753823.1385)(-1,19218640774.1815)(0,6400896725.07333)(1,2133857362.41994)(2,709330404.190125)(3,236564592.561563)(4,81752553.1566176) 
};
\addplot [
color=black,
dashed,
line width=1.0pt,
mark size=2.8pt,
mark=diamond,
mark options={solid},
forget plot
]
coordinates{
 (-5,1.2180e+06)(-4,1193364.65807684)(-3,1198965.75792406)(-2,1193364.65807521)(-1,1062620.60989419)(0,1012080.15988964)(1,941189.168760295)(2,846500.373084454)(3,735816.079234385)(4,633439.8334921) 
};
\addplot [
color=black,
dotted,
line width=1.0pt,
mark size=2.0pt,
mark=o,
mark options={solid},
forget plot
]
coordinates{
 (-1,1.0771e12)(0,248369096113.245)(1,63852246043.6948)(2,19694100910.6242)(3,5389615347.52127)(4,1456031541.24219) 
};
\addplot [
color=black,
dotted,
line width=1.0pt,
mark size=2.8pt,
mark=diamond,
mark options={solid},
forget plot
]
coordinates{
 (-1,1.1297e07)(0,9145148.31690911)(1,8001812.69705405)(2,8086058.47946706)(3,6920279.00238295)(4,5525956.42635132) 
};
\end{semilogyaxis}
\end{tikzpicture}%
\caption{Number of simulation runs for $L=4,5,6$ branch EGC receiver with exponentially correlated Rayleigh fading channels as a function of $\gamma_{th}$. $L=4$ (solid line), $L=5$ (dashed line), and $L=6$ (dotted line). The system parameters are $E_s/N_0=1$ dB, $\sigma=\sqrt{5}$, and $\rho=0.5$.}
\label{fig4}
\end{figure}
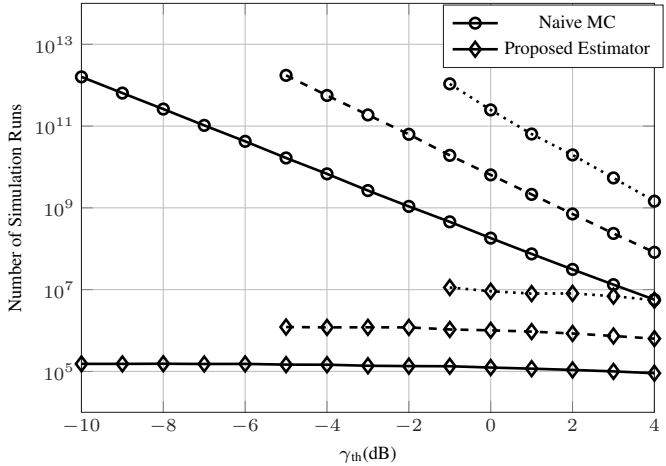
We observe the clear outperformance of our proposed estimator compared to the naive MC sampler. In fact, contrary to the naive MC sampler, which requires a number of runs that keeps increasing as we decrease the OP values, the number of runs needed by our proposed estimator remains bounded, regardless of how much smaller $P_{out}$ is. This is in agreement with the result we have proven in Proposition 2. For example, approximately $10^6$ simulation runs are needed by our proposed IS estimator when $L=5$ and $\gamma_{th}$ is less than $-2$ dB. On the other hand, the naive MC sampler requires approximately $10^{11}$ runs (respectively more than $10^{12}$) for the same value of $L$ and when $\gamma_{th}=-2$ dB (respectively when $\gamma_{th}=-5$ dB). 

\subsection{i.n.i.d Ordered Rayleigh Fading Channels}
In the last experiment, we aim to estimate the OP values at the output of GSC/EGC receivers when operating over i.n.i.d Rayleigh fading channels. In Fig. \ref{fig5}, we plot the values of $P_{out}$ as a function of the threshold for different values of $N$ and $L$. 

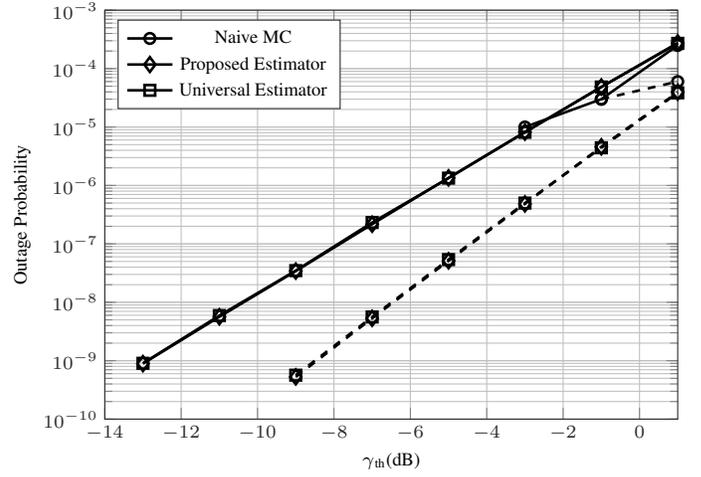
\begin{figure}[h]
\centering
\setlength\figureheight{0.30\textwidth}
\setlength\figurewidth{0.42\textwidth}
%
%
%
%
\begin{tikzpicture}
\scalefont{0.7}
\begin{semilogyaxis}[%
width=\figurewidth,
height=\figureheight,
scale only axis,
every outer x axis line/.append style={darkgray!60!black},
every x tick label/.append style={font=\color{darkgray!60!black}},
xmin=-14, xmax=1,
xminorticks=true,
xlabel={$\gamma{}_{\text{th}}\text{(dB)}$},
xmajorgrids,
xminorgrids,
every outer y axis line/.append style={darkgray!60!black},
every y tick label/.append style={font=\color{darkgray!60!black}},
ymin=1e-10, ymax=0.001,
yminorticks=true,
ylabel={Outage Probability},
ymajorgrids,
yminorgrids,
grid style={solid},
legend style={at={(0.022362040554793,0.758143318430204)},anchor=south west,draw=darkgray!60!black,fill=white,align=left}]
\addplot [
color=black,
solid,
line width=1.0pt,
mark size=2.0pt,
mark=o,
mark options={solid},
]
coordinates{
 (-5,0)(-3,1e-05)(-1,3e-05)(1,0.00025) 
};
\addlegendentry{Naive MC};

\addplot [
color=black,
solid,
line width=1.0pt,
mark size=2.8pt,
mark=diamond,
mark options={solid},
]
coordinates{
 (-13,8.91422814803987e-10)(-11,5.64295875902241e-09)(-9,3.42260694065089e-08)(-7,2.17048623875479e-07)(-5,1.34951160932695e-06)(-3,8.19657515775366e-06)(-1,4.84956967220715e-05)(1,0.000275562229198692) 
};
\addlegendentry{Proposed Estimator};

\addplot [
color=black,
solid,
line width=1.0pt,
mark size=2.0pt,
mark=square,
mark options={solid},
]
coordinates{
 (-13,9.03292451876584e-10)(-11,5.92940716292287e-09)(-9,3.48528678505885e-08)(-7,2.32185045708126e-07)(-5,1.33483458175161e-06)(-3,8.12093723450247e-06)(-1,4.84619796484098e-05)(1,0.000270297934775975) 
};
\addlegendentry{Universal Estimator};
\addplot [
color=black,
dashed,
line width=1.0pt,
mark size=2.0pt,
mark=o,
mark options={solid},
forget plot
]
coordinates{
 (-3,0)(-1,3e-05)(1,6e-05) 
};
\addplot [
color=black,
dashed,
line width=1.0pt,
mark size=2.8pt,
mark=diamond,
mark options={solid},
forget plot
]
coordinates{
 (-9,5.28704321381595e-10)(-7,5.31907639952821e-09)(-5,5.12653081635388e-08)(-3,4.83110831920424e-07)(-1,4.58230747777317e-06)(1,3.94996050396027e-05) 
};
\addplot [
color=black,
dashed,
line width=1.0pt,
mark size=2.0pt,
mark=square,
mark options={solid},
forget plot
]
coordinates{
 (-9,5.64868122309601e-10)(-7,5.60441522032189e-09)(-5,5.38494777173806e-08)(-3,4.99223222308096e-07)(-1,4.41414928532661e-06)(1,3.84582668354301e-05) 
};
\end{semilogyaxis}
\end{tikzpicture}%
\caption{Outage Probability at the output of GSC/EGC receiver with i.n.i.d Rayleigh fading channels as a function of $\gamma_{th}$. $(N,L)=(2,4)$ (solid line) with $\bold{\Omega}=(5,5,8,8)^t$ dB. $(N,L)=(2,5)$ (dashed line) with $\bold{\Omega}=(5,5,5,8,8)^t$ dB. The system parameters are $E_s/N_0=1$ dB, and $M=10^5$.}
\label{fig5}
\end{figure}

Our proposed estimator and the universal estimator yield precise estimates of $P_{out}$ for all values of $\gamma_{th}$ using $10^5$ samples, whereas the failure of the naive MC sampler is evident because it is unable to provide a non-zero estimate when the event is rare. 

\begin{figure}[h]
\centering
\setlength\figureheight{0.30\textwidth}
\setlength\figurewidth{0.42\textwidth}
%
%
%
%
\begin{tikzpicture}
\scalefont{0.7}
\begin{semilogyaxis}[%
width=\figurewidth,
height=\figureheight,
scale only axis,
every outer x axis line/.append style={darkgray!60!black},
every x tick label/.append style={font=\color{darkgray!60!black}},
xmin=-14, xmax=1,
xminorticks=true,
xlabel={$\gamma{}_{\text{th}}\text{(dB)}$},
xmajorgrids,
xminorgrids,
every outer y axis line/.append style={darkgray!60!black},
every y tick label/.append style={font=\color{darkgray!60!black}},
ymin=10000, ymax=10000000000000,
yminorticks=true,
ylabel={Number of Simulation Runs},
ymajorgrids,
yminorgrids,
grid style={solid},
legend style={draw=darkgray!60!black,fill=white,align=left}]
\addplot [
color=black,
solid,
line width=1.0pt,
mark size=2.0pt,
mark=o,
mark options={solid},
]
coordinates{
 (-13,1723806002169.79)(-11,272311044072.74)(-9,44896769452.4287)(-7,7079702414.31236)(-5,1138662250.59958)(-3,187471886.152493)(-1,31684573.7546289)(1,5574844.435405) 
};
\addlegendentry{Naive MC};
\addplot [
color=black,
solid,
line width=1.0pt,
mark size=2.8pt,
mark=diamond,
mark options={solid},
]
coordinates{
 (-13,19773.2235206194)(-11,19447.3043023432)(-9,19877.1596634511)(-7,19131.2043046491)(-5,18453.657714662)(-3,17714.8526083134)(-1,16687.4357927662)(1,15255.557140865) 
};
\addlegendentry{Proposed Estimator};

\addplot [
color=black,
solid,
line width=1.0pt,
mark size=2.0pt,
mark=square,
mark options={solid},
]
coordinates{
 (-13,163872.964283961)(-11,153995.26432508)(-9,159707.34951846)(-7,143021.704610294)(-5,143979.951314661)(-3,130251.969552248)(-1,111286.792985393)(1,90644.2304495018) 
};
\addlegendentry{Universal Estimator};

\addplot [
color=black,
dashed,
line width=1.0pt,
mark size=2.0pt,
mark=o,
mark options={solid},
forget plot
]
coordinates{
 (-9,2906426024988.92)(-7,288892258054.95)(-5,29974264785.8806)(-3,3180717872.8)(-1,335340429.706342)(1,38901130.8286836) 
};
\addplot [
color=black,
dashed,
line width=1.0pt,
mark size=2.8pt,
mark=diamond,
mark options={solid},
forget plot
]
coordinates{
 (-9,41639.9494818232)(-7,39726.9151638824)(-5,38700.4581855239)(-3,37179.7811805164)(-1,33458.9120794342)(1,30370.6404505911) 
};
\addplot [
color=black,
dashed,
line width=1.0pt,
mark size=2.0pt,
mark=square,
mark options={solid},
forget plot
]
coordinates{
 (-9,608247.220249877)(-7,569710.692794717)(-5,528344.505513877)(-3,475685.508159305)(-1,406064.237618579)(1,302150.175968141) 
};
\end{semilogyaxis}
\end{tikzpicture}%
\caption{Number of simulation runs for GSC/EGC receiver with i.n.i.d Rayleigh fading channels as a function of $\gamma_{th}$. $(N,L)=(2,4)$ (solid line) with $\bold{\Omega}=(5,5,8,8)^t$ dB. $(N,L)=(2,5)$ (dashed line) with $\bold{\Omega}=(5,5,5,8,8)^t$ dB. The system parameters are $E_s/N_0=1$ dB.}
\label{fig6}
\end{figure}
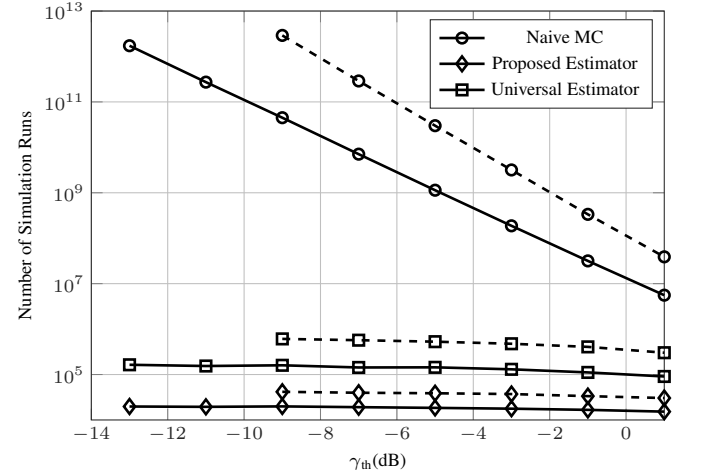
We investigate the efficiency of these estimators in Fig. \ref{fig6} using the necessary number of runs needed in order to obtain $5\%$ relative error. 
As the event of interest becomes rarer and rarer, the number of samples needed by the naive sampler rapidly increases (Fig. \ref{fig6}).. However, the bounded relative error property that our proposed estimator and the universal estimators enjoy is validated in Fig. \ref{fig6}. As expected, our proposed estimator outperforms the universal estimator. Note also that the efficiency of our proposed estimator increases with increasing L, unlike the universal estimator. For example, our estimator requires approximately $8$ (respectively $8 \times 10^7$) times  less number of simulations compared to the universal estimator (respectively the naive MC sampler) when $(N,L)=(2,4)$ and $\gamma_{th}=-13$ dB. However, when $(N,L)=(2,5)$ and $\gamma_{th}=-9$ dB, our proposed estimator is approximately $15$ times more efficient than the universal estimator.

\section{Conclusion}
We developed an importance sampling estimator for the estimation of the outage probability at the output of equal gain combining receivers. Our proposed biased probability density function is the truncation of the underlying one over the multidimensional sphere with a radius given by the specified threshold. Our method is based on the perfect knowledge of a closed-form expression of the outage probability with maximum ratio combining receivers. This assumption is not restrictive since it holds for various challenging fading models. We extended our approach to the case of generalised selection combining receivers combined with equal gain combining technique for  independent Rayleigh fading channels. We proved  that our proposed estimator has bounded relative error for four interesting fading channels. This study represents a valuable contribution to the field of variance reduction techniques. Finally, we tested the performance of our proposed estimator through various simulations. 
\bibliography{References}
\bibliographystyle{IEEEtran}
\end{document}